\documentclass[12pt]{amsart}

\usepackage[utf8]{inputenc}
\usepackage[OT1]{fontenc}
\usepackage{amsfonts, amsmath, amssymb}
\usepackage{graphicx}
\usepackage{listings}
\usepackage{xcolor}
\usepackage{hyperref}
\usepackage{xcolor}
\usepackage{multirow}
\usepackage{tabularx,booktabs}
\usepackage{lipsum}
\usepackage{diagbox}
\usepackage{tablefootnote}

\newtheorem{theorem}{Theorem}[section]
\newtheorem{lemma}[theorem]{Lemma}

\newtheorem{corollary}[theorem]{Corollary}

\theoremstyle{definition}
\newtheorem{definition}[theorem]{Definition}
\newtheorem{example}[theorem]{Example}

\theoremstyle{remark}
\newtheorem{remark}[theorem]{Remark}

\providecommand{\x}{\boldsymbol{x}}
\providecommand{\y}{\boldsymbol{y}}

\def\qi#1 {\fbox {\footnote {\ }}\ \footnotetext { From Qi: {\color{red}#1}}}

\providecommand{\E}[1]{\mathbb{E} \left[ #1 \right]}

\begin{document}
\title[On the size distribution of FLL balls with radius one]{On the size distribution of the fixed-length Levenshtein balls with radius one}

\author{Geyang~Wang}
\address{Department of Electronic and Computer Engineering, Institute for Systems Research, University of Maryland, College Park, MD 20742, USA} 
\curraddr{}
\email{wanggy@umd.edu}
\thanks{}

\author{Qi~Wang}
\address{Department of Computer Science and Engineering, National Center for Applied Mathematics Shenzhen, Southern University of Science and Technology, Nanshan District, Shenzhen, Guangdong 518055, China}
\email{wangqi@sustech.edu.cn}
\thanks{The paper was presented (in part) at The Twelfth International Workshop on Coding and Cryptography (WCC), Rostock, Germany, March 7--11, 2022.}



\date{}

\dedicatory{}

\maketitle              
\begin{abstract}
  The \emph{fixed-length Levenshtein} (FLL) distance between two words $\x, \y \in \mathbb{Z}_m^n$ is the smallest integer $t$ such that $\x$ can be transformed to $\y$ by $t$ insertions and $t$ deletions. The size of a ball in the FLL metric is a fundamental yet challenging problem. Very recently, Bar-Lev, Etzion, and Yaakobi explicitly determined the minimum, maximum and average sizes of the FLL balls with radius one, respectively. In this paper, based on these results, we further prove that the size of the FLL balls with radius one is highly concentrated around its mean by Azuma's inequality.
\end{abstract}
\section{Introduction}\label{sec-intro}

The \emph{Levenshtein distance} (also known as {\em edit distance}) between two words is the smallest number of deletions and insertions needed to transform one word to the other. This is a metric used for codes correcting synchronization errors. The theory of coding with respect to the Levenshtein distance dates back to the 1960s~\cite{Leven66}, but there has been much less progress in comparison with the classical theory of coding with respect to the Hamming distance. As commented by Mitzenmacher~\cite{Mitzen09}, \textit{``Channels with synchronization errors, including both insertions and deletions as well as more general timing errors, are simply not adequately
understood by current theory. Given the near-complete knowledge we have for channels with erasures and errors \ldots, our lack of understanding about channels with synchronization errors is truly remarkable.''} Indeed, even the fundamental problem of counting the number of words formed by deleting and inserting symbols into a given word remains elusive.

In 1966, Levenshtein~\cite{Leven66} gave the earliest bound on the number of words formed by deleting a constant number of symbols from a word. This bound was later improved by Calabi and Hartnett~\cite{Calabi69}; also by Hirschberg and Regnier~\cite{Hirs00}.  On the other hand, the number of words formed by inserting $r$ symbols into $\x \in \mathbb{Z}_m^n$ does not depend on $\x$ itself, and is given by $\sum_{i=0}^{r}\binom{n+r}{i}(m-1)^i$~\cite{Leven66}. Motivated by
estimating the rate of synchronization error-correction codes, Sala and Dolecek~\cite{Sala13} studied the number of words formed by deleting and inserting a constant number of symbols into a given word. The \emph{fixed-length Levenshtein (FLL) distance} (originally under the name ``\emph{ancestor distance}'') between two words $\x, \y \in \mathbb{Z}_m^n$ is defined as half of the traditional Levenshtein distance, i.e., the smallest $t$ such that $\x$ can be
transformed to $\y$ by $t$ insertions \emph{and} $t$ deletions. An explicit expression was given on the FLL ball size with radius one (see Lemma~\ref{lemma: sala}), and the size of the FLL balls with radius larger than one was upper bounded in~\cite{Sala13}. Very recently, Bar-Lev, Etzion, and Yaakobi~\cite{Barlev22} found the explicit expressions for the minimum, the maximum and the average sizes of {the} FLL balls with radius one, respectively. A natural follow-up question is how the size of the FLL balls with radius one is distributed. In this paper, we prove that the size of the FLL balls with radius one is highly concentrated around its mean by Azuma's inequality (see Theorem~\ref{thm: binary distribution} and Theorem~\ref{thm: m-ary distribution}).

The rest of this paper is organized as follows.
In Section~\ref{sec: preliminaries}, we provide some notations, definitions, and auxiliary results. In Section~\ref{sec: distribution}, we analyze in detail the size distribution of the FLL balls with radius one, and state the main results in Theorem~\ref{thm: binary distribution} and Theorem~\ref{thm: m-ary distribution}.
Finally, we conclude this paper in Section~\ref{sec: conclusion}.

\section{Preliminaries}~\label{sec: preliminaries}

Let $\mathbb{Z}_m = \{0,1,\dots,m-1\}$ for an integer $m \ge 2$. We use $[n]$ to denote the set $\{1,\dots,n\}$. In this paper, vectors in $\mathbb{Z}_m^n$ are called $m$-ary sequences (words) with length $n$, and {are written} as strings for convenience. For a word $\x = x_1, \dots, x_n \in \mathbb{Z}_m^n$ and $1 \leq i \le j \le n$, the subsequence $x_i, x_{i+1}, \dots, x_j$ is denoted by $\x_{[i,j]}$. A \emph{run} of $\x$ is a {\em maximal} subsequence $\x_{[i,j]}$ with all symbols from $x_i$ to $x_j$ identical, i.e., $x_i = x_{i+1} = \dots = x_j$, $x_{i-1} \ne x_i$ for $i \ge 2$ and~$x_{j+1} \ne x_j$ for $j \leq n- 1$. An \emph{alternating segment} of $\x$ is a \emph{maximal} subsequence $\x_{[i,j]}$ with the form $abab\dots ab$ or $abab\dots ba,$ where $a,b \in \mathbb{Z}_m$ and $a \ne b$. Note that by `maximal' we mean that both $\x_{[i-1,j]}$ (when $i>1$) and $\x_{[i,j+1]}$ (when $j<n$) are no longer alternating. The number of runs in $\x$ is denoted by $\rho(\x)$, and the number of alternating segments in $\x$ is denoted by $a(\x)$. For each $\x \in \mathbb{Z}_2^n$, we have $\rho(\x) + a(\x) = n + 1$. The lengths of the first and the last alternating segments in a word $\x$ are of particular interest in this work and denoted by {$h(\x)$ and $t(\x)$}, respectively. 

\begin{example}
  Let $\x = 001100101$. Then the runs of $\x$ are $00, 11, 00, 1, 0, 1$ and $\rho(\x) = 6$; The alternating segments of $\x$ are $0, 01, 10, 0101$ and $a(\x) = 4$, $h(\x) = 1$, $t(\x) = 4$.
\end{example}

\begin{lemma}\label{lemma: exp h and t}
  Let $n>0$, $m>1$ both be integers. If $\x$ is picked uniformly at random, then we have
    \begin{equation}\label{equ: exp h}
        \underset{\x \in \mathbb{Z}_m^n}{\mathbb{E}}\left[ h(\x) \right] = \underset{\x \in \mathbb{Z}_m^n}{\mathbb{E}}\left[ h(\x) \vert x_1 \right] = 2 - \frac{1}{m^{n-1}},
    \end{equation}
    \begin{equation}\label{equ: exp t}
        \underset{\x \in \mathbb{Z}_m^n}{\mathbb{E}}\left[ t(\x) \right] = \underset{\x \in \mathbb{Z}_m^n}{\mathbb{E}}\left[ t(\x) \vert x_n \right] = 2 - \frac{1}{m^{n-1}}.
    \end{equation}
\end{lemma}

\begin{proof}
    We prove Eq.~\eqref{equ: exp h}, and Eq.~\eqref{equ: exp t} then follows by symmetry.
    Let $\x \in \mathbb{Z}_m^n$ be an arbitrary sequence.
    Define the indicator variable $I_i$ for $i \in [n]$ as follows.
    \[
        I_i = \begin{cases}
            1, & \text{ if $x_i$ belongs to the first alternating segment}, \\
            0, & \text{ otherwise.}
        \end{cases}    
    \]
    Then we have $h(\x) = \sum_{i=1}^{n}I_i$. Also note that $\Pr(I_1 = 1) = 1$, and $ \Pr(I_i = 1) = (m-1)/m^{i-1}$ for $i > 1$. Therefore, we have
    \[
        \begin{aligned}
        \E{h(\x)} & = \sum_{i=1}^n\E{I_i} = 1 + (m-1)\left(\frac{1}{m} + \dots + \frac{1}{m^{n-1}}\right) \\
        & = 2 - \frac{1}{m^{n-1}}.
        \end{aligned}
    \]

    Intuitively, exposing the first symbol of a random $\x \in \mathbb{Z}_m^n$ does not provide any information about $h(\x)$ since $h(\x) \ge 1$ for all words. 
    Thus, we have $\E{h(\x)} = \E{h(\x) \middle \vert x_1}$.
    More precisely, we define an equivalence relation ``$\approx$'' on $\mathbb{Z}_m^n$ as follows.
    For $\x, \y \in \mathbb{Z}_m^n$, we say that $\x \approx \y$ if and only if $x_i \equiv y_i + k \mod{m}$ for each $i \in [n]$ and some integer $k$. 
    Clearly, $h(\x) = h(\y)$ if $\x \approx \y$, and $\{a\} \times \mathbb{Z}_m^{n-1}$ contains exactly one element from each equivalence class. 
    It then follows that $\E{h(\x) \vert x_1 = a} = \E{h(\x) \vert x_1 = b}$.
    Thus, we have
    \[
    \E{h(\x)} = \sum_{a \in \mathbb{Z}_m} \frac{1}{m} \E{h(\x) \vert x_1 = a} = \E{h(\x) \vert x_1 = a},
    \]
    and the proof is then completed. 

\end{proof}

A sequence $\y \in \mathbb{Z}_m^{n-t}$ with $t \in [n-1]$ is called a \emph{$t$-subsequence} of $\x \in \mathbb{Z}_m^n$ if $\y$ is formed by deleting $t$ symbols from $\x$. In other words, $\y = (x_{i_1}, x_{i_2}, \dots, x_{i_{n-t}})$, where $1 \leq i_1 < i_2 < \dots <i_{n-t} \leq n$. Likewise, $\x$ is called a \emph{$t$-supersequence} of $\y$. The set of all $t$-subsequences of $\x$ is called the {\em deletion $t$-sphere} centered at $\x$, and denoted by $D_t(\x)$. The set of all $t$-supersequences of $\x$ is called the {\em insertion $t$-sphere} centered at $\x$, and denoted by $I_t(\x)$. The {fixed-length} Levenshtein distance is formally defined as follows.

\begin{definition}[Fixed-length Levenshtein (FLL) distance]
  The fixed-length Levenshtein (FLL) distance between two words $\x, \y \in \mathbb{Z}_m^n$ is the smallest $t$ such that $D_t(\x) \cap D_t(\y) \ne \emptyset$, and is denoted by $d_l(\x,\y)$.
\end{definition}

It is easy to see that $d_l(\x,\y) = t$ if and only if $t$ is the smallest integer such that $\x$ can be transformed to $\y$ by $t$ deletions and $t$ insertions.

\begin{definition}[FLL ball]
For each word $\x \in \mathbb{Z}_m^n$, the \emph{FLL $t$-ball} centered at $\x$ is defined by
    \begin{equation*}
        L_t(\x) \triangleq \{\y \in \mathbb{Z}_m^n  \vert \ d_l(\x,\y) \le t \},
    \end{equation*}
    and $t$ is called the radius.
\end{definition}

The following results on the size of the FLL ball were given in~\cite{Sala13,Barlev22}, and will be useful later. 

\begin{lemma}~\cite[Theorem 1]{Sala13}\label{lemma: sala}
    For all $\x \in \mathbb{Z}_m^n$,
    \begin{equation}\label{equ: ball size}
        |L_1(\x)| = \rho(\x) (mn - n - 1) + 2 - \frac{1}{2}\sum_{i=1}^{a(\x)} s_i^2 + \frac{3}{2}\sum_{i=1}^{a(\x)}s_i - a(\x),
    \end{equation}
    where $s_i$, for $1 \le i \le a(\x)$, is the length of the $i$-th alternating segment of $\x$.
\end{lemma}

Lemma~\ref{lemma: sala} was proved by a careful argument based on the  principle of inclusion and exclusion. Moreover, it can be shown that $\rho(\x)(mn - n - 1) + 2$ is an upper bound of $|L_1(\x)|$, and Eq.~\eqref{equ: ball size} was obtained by subtracting the number of overcounted sequences.

\begin{lemma}~\cite[Lemma 11, Corollary 8, Lemma 12, and Lemma 16]{Barlev22} \label{lemma: partial expectation}
    Let the notations be the same as above. For integers $m, n > 1$, we have
    \begin{align}
        & \underset{\x \in \mathbb{Z}_{m}^{n}}{\mathbb{E}}\left[\sum_{i=1}^{a(\x)} s_{i}\right]=n+(n-2) \cdot \frac{(m-1)(m-2)}{m^{2}},  \nonumber \\
        & \underset{\x \in \mathbb{Z}_{m}^{n}}{\mathbb{E}}[a(\x)]=1+\frac{(n-2)(m-1)(m-2)}{m^{2}}+\frac{n-1}{m}, \nonumber \\
        & \underset{\x \in \mathbb{Z}_{m}^{n}}{\mathbb{E}}[\rho(\x)]=n-\frac{n-1}{m}, \nonumber \\
	& \underset{\x \in \mathbb{Z}_{m}^{n}}{\mathbb{E}}\left[\sum_{i=1}^{a(\x)} s_{i}^{2}\right]=\frac{n\left(4 m^{2}-3 m+2\right)}{m^{2}}+\frac{6 m-4}{m^{2}}-4-\frac{2}{m-1}\left(1-\frac{1}{m^{n}}\right) + \frac{2}{m^n}. \label{eq: ES2}
    \end{align}
\end{lemma}

\begin{remark}
  Eq.~\eqref{eq: ES2} was proved in~\cite[Lemma 16]{Barlev22}. However, the last term $\frac{2}{m^n}$ was missing by mistake during the calculation\footnote{In~\cite{Barlev22}, on Page 2334, left column, from the fourth equation to the fifth equation, the last term of the fifth equation should not be $\frac{2}{q-1}$ but $\frac{2q}{q-1}$. This error spreads out in the following calculations therein.}. Theorem~\ref{thm: barlev22} has also been corrected accordingly.
\end{remark}

\begin{theorem}~\cite[Theorem 13]{Barlev22}\label{thm: barlev22}
    For integers $m,n > 1$, we have 
    \begin{equation}\label{equ: expectation}
      \underset{\x \in \mathbb{Z}_{m}^{n}}{\mathbb{E}}[|L_1(\x)|] = n^2 \left(m + \frac{1}{m} - 2\right) + 2 - \frac{n}{m} + \frac{m^{n-1} - 1}{m^{n-1}(m - 1)}.
    \end{equation}
\end{theorem}

 Note that the proof follows from Lemma~\ref{lemma: sala}, Lemma~\ref{lemma: partial expectation} and the linearity of expectation.


A {\em martingale} is a sequence of real random variables $Z_0, \dots, Z_n$ with finite expectation such that for each $0 \le i < n$, 
\[
  \E{Z_{i+1} \middle \vert Z_i, Z_{i-1}, \dots, Z_0} = Z_i.  
\]

A classical martingale named \emph{Doob martingale} (see {for example}~\cite{AS16}) will be used in this paper. Let $X_1, \dots, X_n$ be {the} underlying random variables (not necessarily independent) and $f$ be a function over $X_1, \dots, X_n$. The Doob martingale $Z_0, \dots, Z_n$ is defined by
\begin{equation*}
    \begin{aligned}
        Z_0 & = \E{f(X_1, \dots, X_n)}; \\
        Z_i & = \E{f(X_1, \dots, X_n) \middle \vert X_1, \dots, X_i} \text{ for } i \in [n].
    \end{aligned}
\end{equation*}
In other words, $Z_i$ is defined by the expected value of $f$ after exposing $X_1, \dots, X_i$.
The following classical {result~\cite{AS16}} plays a key role in {the proof of our main results.}

\begin{theorem}[Azuma's inequality]\label{thm: azuma}
    Let $Z_0,Z_1,\dots,Z_n$ be a martingale such that for each $1 \le i \le n$,
    \begin{equation*}
        |Z_i - Z_{i-1}| \le c_i.
    \end{equation*}
    Then for every $\lambda > 0$, {we have}
    \begin{equation*}
        \Pr(Z_n - Z_0 \ge \lambda) \le \exp \left( \frac{-\lambda^2}{2(c_1^2 + \dots + c_n^2)}\right),
    \end{equation*}
    and 
    \begin{equation*}
        \Pr(Z_n - Z_0 \le -\lambda) \le \exp \left( \frac{-\lambda^2}{2(c_1^2 + \dots + c_n^2)}\right).
    \end{equation*}
\end{theorem}

\section{The size distribution {of the FLL balls with radius one}}\label{sec: distribution}

In this section, we discuss the size distribution of the FLL balls with radius one. We start with the binary case and then deal with the general $m$-ary case.

\subsection{The binary case} \label{sec: binary case}

Let $n, n'$ be positive integers. In order to estimate $|L_1(\x)|$, we define the following notation: for each $\y \in \mathbb{Z}_2^{n'}$, define 

\begin{equation}\label{equ: f_n}
    f_n(\y) = \rho(\y)n - \frac{1}{2}\sum_{i = 1}^{a(\y)}s_i^2,
\end{equation}
where $s_i$ is the length of {the} $i$-th alternating segment of $\y$ for $1 \le i \le a(\y)$.
Note that $\sum_{i=1}^{a(\y)}s_i = n'$ for all $\y \in \mathbb{Z}_2^{n'}$, and $\rho(\y) + a(\y) = n'+1$.
Then by Eq.~\eqref{equ: ball size}, we have $|L_1(\x)| = f_n(\x) + \frac{n}{2} + 1$ for each $\x \in \mathbb{Z}_2^n$.
Therefore, it {suffices} to find the distribution of $f_n(\x)$ for $\x \in \mathbb{Z}_2^n$. 
To this end, we express {$f_n(\y)$} by two partial values $f_n(\y_{[1,i]})$ and $f_n(\y_{[i+1,n']})$ in the following lemma.

\begin{lemma}\label{lemma: f_n partition}
  Let {$n, n'$ be positive} integers.
    For each $i \in [n'-1]$ and $\y \in \mathbb{Z}_2^{n'}$, we have 

    \begin{equation*}
        f_n(\y) = 
        \begin{cases}
            f_n(\y_{[1,i]}) + f_n(\y_{[i+1,n']}) - n ,  & \text{if $y_i = y_{i+1}$},\\

            f_n(\y_{[1,i]}) + f_n(\y_{[i+1,n']}) - t(\y_{[1,i]}) h(\y_{[i+1, n']}) ,  & \text{if $y_i \ne y_{i+1}$}.
        \end{cases}
    \end{equation*}
\end{lemma}

\begin{proof}
    It is important to use the additive behavior of $f_n(\y)$ in Eq.~\eqref{equ: f_n}. 
    The first case follows from the equations of $\rho(\y) = \rho(\y_{[1,i]}) + \rho(\y_{[i+1,n']}) - 1$ and $a(\y) = a(\y_{[1,i]}) + a(\y_{[i+1,n']})$. 
    
    When $y_{i} \ne y_{i+1}$, we have $\rho(\y) = \rho(\y_{[1,i]}) + \rho(\y_{[i+1,n']})$. Then the difference of $f_n(\y)$ and $f_n(\y_{[1,i]}) + f_n(\y_{[i+1,n']})$ is given by 
    \[
        \frac{1}{2} \left[ \left(t(\y_{[1,i]}) + h(\y_{[i+1, n']})\right)^2  - t(\y_{[1,i]})^2 - h(\y_{[i+1, n']})^2 \right] = t(\y_{[1,i]}) h(\y_{[i+1, n']}).
    \]
    The proof is then completed.

\end{proof}

\begin{corollary}\label{corollary: f difference}
    Let $n, n'$ be positive integers.
    For all $\y \in \mathbb{Z}_2^{n'}$, we have 
    \begin{equation*}
        f_n(\y) - f_n(\y_{[1,n'-1]}) = 
        \begin{cases}
	  -\frac{1}{2},  & \text{if } y_{n'-1} = y_{n'}, \\
	  n - \frac{1}{2} - t(\y_{[1, n'-1]}),  & \text{if } y_{n'-1} \ne y_{n'}.
        \end{cases}
    \end{equation*}
\end{corollary}

\begin{proof}
  Note that by the definition of $f_n(\y)$ in Eq.~\eqref{equ: f_n}, we have $f_n(y_{n'}) = n - \frac{1}{2}$. By setting $i = n'-1$, the result then follows from Lemma~\ref{lemma: f_n partition}. 
\end{proof}

Now we are ready to discuss the distribution of $|L_1(\x)|$ for uniformly distributed $\x \in \mathbb{Z}_2^n$. Since the case that $n \le 3$ is trivial, we now focus on the cases that $n$ is large.

\begin{theorem}\label{thm: binary distribution}
    Let $n > 3$ be an integer and $x_1, \dots, x_n$ be independent random variables such that $\Pr(x_i = 0) = \Pr(x_i = 1) = \frac{1}{2}$ for $i \in [n]$. Then for the word $\x = x_1\cdots x_n$, we have
    \begin{equation}\label{equ: upper tail}
        \Pr \left( {|L_1(\x)| - \underset{\x \in \mathbb{Z}_2^n}{\mathbb{E}} \left[ |L_1(\x)| \right] \ge c n \sqrt{n-1}}  \right) \le e^{-2c^2},
    \end{equation}
    and 
    \begin{equation}\label{equ: lower tail}
        \Pr \left( {|L_1(\x)| - \underset{\x \in \mathbb{Z}_2^n}{\mathbb{E}} \left[ |L_1(\x)| \right] \le - c n \sqrt{n-1}}  \right) \le e^{-2c^2},
    \end{equation}
    where $\underset{\x \in \mathbb{Z}_2^n}{\mathbb{E}} \left[ |L_1(\x)| \right] = \frac{n^2}{2} - \frac{n}{2} - \frac{1}{2^{n-1}} + 3$, and $c$ is a positive constant.
\end{theorem}

\begin{proof}
    We define the Doob martingale $Z_0 = \E{f_n(\x)}$, and $Z_i = \E{f_n(\x) \vert \x_{[1,i]}}$ by exposing one $x_i$ at a time for $i \in [n]$.
    Clearly, by Eq.~\eqref{equ: f_n} and Lemma~\ref{lemma: partial expectation}, we have $Z_n = f_n(\x)$, and
    \begin{eqnarray*}\label{eqn-z0}
     Z_0  & = & \E{\rho(\x) n - \frac{1}{2} \sum_{i=1}^{a(\x)} s_i^2 } \\
      & = & n \left(n - \frac{n-1}{2}\right) \\
      &  & \quad - \frac{1}{2} \left[ \frac{n (4\cdot 2^2 - 3 \cdot 2 +2)}{2^2} +\frac{6 \cdot 2 - 4}{2^2} - 4 - 2 \left(1 - \frac{1}{2^n}\right) + \frac{2}{2^n} \right] \\ 
      & = & \frac{n^2}{2}- n + 2 - \frac{1}{2^{n-1}}. 
    \end{eqnarray*}
    
For every $1 \le i < n$, we have
    \begin{eqnarray*}
   Z_i & = & \E{f_n(\x) \middle\vert \x_{[1,i]}} \\
       & = & \E{f_n(\x) \middle\vert \x_{[1,i]}, x_{i} = x_{i+1}} \Pr(x_i = x_{i+1}) \\
        &  & \quad + \E{f_n(\x) \middle\vert \x_{[1,i]}, x_{i} \ne x_{i+1}} \Pr(x_i \ne x_{i+1}) .
    \end{eqnarray*}
Then by Lemma~\ref{lemma: f_n partition}, we have
    \begin{eqnarray*}
      \lefteqn{Z_i = \frac{1}{2} \E{ f_n(\x_{[1,i]}) + f_n(\x_{[i + 1,n]}) -n \middle\vert \x_{[1,i]}, x_i = x_{i+1}}} \\
    &  & \quad + \frac{1}{2} \E{ f_n(\x_{[1,i]}) + f_n(\x_{[i + 1,n]}) -t(\x_{[1,i]})h(\x_{[i + 1,n]}) \middle\vert \x_{[1,i]}, x_i \ne x_{i+1}} \\
    & = &  f_n(\x_{[1,i]}) + \E{f_n(\x_{[i + 1,n]})}  - \frac{n}{2} - \frac{1}{2}t(\x_{[1,i]})\E{h(\x_{[i + 1,n]}) \middle\vert x_i \ne x_{i+1}}.
    \end{eqnarray*}
    
    Note that by Eq.~\eqref{equ: f_n} and Lemma~\ref{lemma: partial expectation}, we have
    \begin{eqnarray*}
     \E{f_n(\x_{[i + 1,n]})} & = & \E{\rho(\x_{[i+1,n]}) n - \frac{1}{2} \sum_{i=1}^{a(\x_{[i+1,n]})} s_i^2 } \\
       & = &  n \cdot \frac{n-i+1}{2} - \frac{1}{2}\left[3(n-i) - 4 + \frac{1}{2^{n-i-2}}\right] \\
        & = & \frac{n^2}{2} - \frac{in}{2} + 2 - \frac{3}{2}(n - i) + \frac{n}{2} - \frac{1}{2^{n-i-1}} , 
    \end{eqnarray*}
    and by Lemma~\ref{lemma: exp h and t}, we have
    \[
        \E{h(\x_{[i + 1,n]}) \middle\vert x_i \ne x_{i+1}} = 2 - \frac{1}{2^{n-i-1}}.
    \]
    It then follows that
    \begin{equation}\label{eq:Zi}
        Z_i =  f_n(\x_{[1,i]}) + \frac{n^2}{2} - \frac{in}{2} + 2 - \frac{3}{2}(n - i) - \frac{1}{2^{n-i-1}} - t(\x_{[1,i]})\left(1 - \frac{1}{2^{n-i}}\right).
    \end{equation}
    Notice that $Z_n = f_n(\x)$ also satisfies Eq.~\eqref{eq:Zi}. Thus, for $ 1 < i \le n$, we have
    \begin{eqnarray*}
     Z_i - Z_{i-1} & = & f_n(\x_{[1,i]}) - f_n(\x_{[1,i-1]}) - \frac{n}{2} + \frac{3}{2} - \frac{1}{2^{n-i}} \\
        & &  \quad - t(\x_{[1,i]})\left(1-\frac{1}{2^{n-i}}\right) + t(\x_{[1,i-1]})\left(1 - \frac{1}{2^{n-i+1}}\right).
    \end{eqnarray*}
    The difference between $f_n(\x_{[1,i]}) $ and $f_n(\x_{[1,i-1]})$ is given by Corollary~\ref{corollary: f difference}, also note that $t(\x_{[1,i]}) = 1$ if $x_{i-1} = x_i$ and $t(\x_{[1,i]}) = t(\x_{[1,i-1]}) + 1$ if $x_{i-1} \ne x_i$. Therefore, we have
    \begin{equation*}\label{eq: Z_diff}
        Z_i - Z_{i-1} =
        \begin{cases}
	  \frac{n}{2} - t(\x_{[1, i-1]})\left( 1 - \frac{1}{2^{n-i+1}} \right), & \textrm{if } x_{i-1} \ne x_i, \\
	  - \frac{n}{2} + t(\x_{[1, i-1]})\left(1 - \frac{1}{2^{n-i+1}}\right), & \textrm{if } x_{i-1} = x_i. \\
        \end{cases}
    \end{equation*}
    It is straightforward to check that $|Z_i - Z_{i-1}| \le \frac{n}{2}$ for $1 < i \le n$ and $Z_1 - Z_0 = 0$.\footnote{It can be verified by taking $i=1$ in Eq.~\eqref{eq:Zi}. Alternatively, it also follows from the symmetry induced by the equivalence relation defined in the proof of Lemma~\ref{lemma: exp h and t}.}
    Then by Theorem~\ref{thm: azuma}, we have
    \begin{equation*}
        \Pr(Z_n - Z_0 \ge \lambda) \le \exp \left( \frac{-\lambda^2}{2(0^2 + (n-1) (\frac{n}{2})^2)}\right) = \exp \left( \frac{-2\lambda^2}{n^2(n-1)} \right).
    \end{equation*}
    Take $\lambda = cn\sqrt{n-1}$, where $c$ is a positive constant, then we have
    \begin{equation}\label{equ: z upper tail}
        \Pr(Z_n - Z_0 \ge cn \sqrt{n-1}) \le e^{-2c^2} ,
    \end{equation}
  and
    \begin{equation} \label{equ: z lower tail}
        \Pr(Z_n - Z_0 \le - cn \sqrt{n-1}) \le e^{-2c^2}.
    \end{equation}

    Note that $|L_1(\x)| = Z_n + \frac{n}{2} + 1$, and $\underset{\x \in \mathbb{Z}_2^n}{\mathbb{E}} \left[ |L_1(\x)| \right] = Z_0 + \frac{n}{2} + 1$.
    Then Eq.s~\eqref{equ: upper tail} and~\eqref{equ: lower tail} {follow} from Eq.s~\eqref{equ: z upper tail} and~\eqref{equ: z lower tail}, respectively. The proof is completed.
\end{proof}

\subsection{The $m$-ary case}

Let $n,n'$ be positive integers. Again, to analyze $|L_1(\x)|$, for each $\y \in \mathbb{Z}_m^{n'}$, we define
\begin{equation*}
    f_{m,n}(\y) = \rho(\y)(mn-n-1) - \frac{1}{2}\sum_{i=1}^{a(\y)}s_i^2 + \frac{3}{2}\sum_{i=1}^{a(\y)}s_i  - a(\y),
\end{equation*}
where $s_i$ is the length of the $i$-th alternating segment of $\y$, and then $|L_1(\x)| = f_{m,n}(\x) + 2$.
In parallel to Lemma~\ref{lemma: f_n partition}, we try to express $f_{m,n}(\y)$ by $f_{m,n}(\y_{[1,i]})$ and $f_{m,n}(\y_{[i+1,n']})$ as follows.

\begin{lemma}\label{lemma: f_mn partition}
    Let $n > 1$ and $m, n' > 2$. For each $i \in [n'-1]$ and $\y \in \mathbb{Z}_m^{n'}$, we have
    \begin{itemize}
        \item If $y_i = y_{i+1}$,
            \begin{equation}\label{equ: f_mn partition 1}
                f_{m,n}(\y) = f_{m,n}(\y_{[1,i]}) + f_{m,n}(\y_{[i+1,n']}) - mn + n + 1;
            \end{equation}
        \item If $y_i \ne y_{i+1}$,
            \begin{equation}\label{equ: f_mn partition bound}
                \begin{aligned}
                    f_{m,n}(\y) & \le  f_{m,n}(\y_{[1,i]}) + f_{m,n}(\y_{[i+1,n']}), \\
                    f_{m,n}(\y) & \ge  f_{m,n}(\y_{[1,i]}) + f_{m,n}(\y_{[i+1,n']}) + 1 - t(\y_{[1,i]})h(\y_{[i+1,n']}).
                \end{aligned}
            \end{equation}
    \end{itemize}
\end{lemma}

\begin{proof}
    Let $\{s_i\}_{i=1}^{a(\x)}$ be the lengths of maximal alternating segments in $\y$.
    To lighten the notations, denote by $\boldsymbol{u} := \y_{[1,i]}$ and $\boldsymbol{v} := \y_{[i+1, n']}$.
    If $y_i = y_{i+1}$, then $y_i$ and $y_{i+1}$ belong to different alternating segments, and
    \begin{align*}
        \rho(\y) &= \rho(\boldsymbol{u}) + \rho(\boldsymbol{v}) - 1, \\
        a(\y) &= a(\boldsymbol{u}) + a(\boldsymbol{v}).
    \end{align*}
    Therefore, Eq.~\eqref{equ: f_mn partition 1} follows from the additive behavior of $f_{m,n}(\y)$.

    It remains to discuss the case when $y_i \ne y_{i+1}$.
    Observe that in this case, we split exactly one alternating segment $\y_{[i-l+1,i+r]}$ into two segments: $\y_{[i-l+1, i]}$ and $\y_{[i+1, i+r]}$, where $l,r \ge 1$.
    In addition, $a(\y) = a(\boldsymbol{u}) + a(\boldsymbol{v}) - 1$.
    The difference between $f_{m,n}(\y)$ and $f_{m,n}(\boldsymbol{u}) + f_{m,n}(\boldsymbol{v})$ is given by
    \begin{equation*}
        \begin{split}
            - \frac{1}{2} (l+r)^2 + \frac{1}{2}(l^2 + r^2) + 1 = 1 - lr,
        \end{split}
    \end{equation*}
    and
    \begin{equation*}
        f_{m,n}(\y) = f_{m,n}(\y_{[1,i]}) + f_{m,n}(\y_{[i+1,n']}) + 1 - lr.
    \end{equation*}
    Note that if $l,r >1$, then $l = t(\boldsymbol{u})$, $r = h(\boldsymbol{v})$. Thus, Eq.~\eqref{equ: f_mn partition bound} holds.
\end{proof}

\begin{theorem}\label{thm: m-ary distribution}
    Let $m>2, n>3$ be integers, and $x_1, \dots, x_n$ be independent random variables such that $\Pr(x_i = j) = \frac{1}{m}$ for $i \in [n]$, $j \in \mathbb{Z}_m$. Then for the word $\x = x_1,\dots,x_n$, we have
    \begin{equation}\label{equ: m-ary lower tail}
        \Pr \left( {|L_1(\x)| - \underset{\x \in \mathbb{Z}_m^n}{\mathbb{E}} \left[ |L_1(\x)| \right] \ge  c (m + \frac{1}{m}) n \sqrt{n-1}}  \right) \le e^{-c^2/2},
    \end{equation}
    and
    \begin{equation}\label{equ: m-ary upper tail}
        \Pr \left( {|L_1(\x)| - \underset{\x \in \mathbb{Z}_m^n}{\mathbb{E}} \left[ |L_1(\x)| \right] \le - c (m + \frac{1}{m}) n \sqrt{n-1}}  \right) \le e^{-c^2/2},
    \end{equation}
    where $\underset{\x \in \mathbb{Z}_m^n}{\mathbb{E}} \left[ |L_1(\x)| \right]$ is given in Eq.~\eqref{equ: expectation}, and {$c$} is a positive constant.
\end{theorem}

\begin{proof}
    As in Section~\ref{sec: binary case}, define the Doob martingale $Z_0 = \E{f_{m,n}(\x)}$, and $Z_i = \E{f_{m,n}(\x) \middle \vert \x_{[1,i]}}$ for $1 \le i \le n$.
    Note that by symmetry defined in the proof of Lemma~\ref{lemma: exp h and t},
    \[
        Z_1 = Z_0 = n^2 \left(m + \frac{1}{m} -2\right)  - \frac{n}{m} - \frac{1}{m^n(m-1)} + \frac{1}{m-1} - \frac{1}{m^n}.
    \]

For $1 < i < n$,
\begin{eqnarray*}
Z_i & = & \E{f_{m,n}(\x) \middle\vert \x_{[1,i]}} \\
    & = & \E{f_n(\x) \middle\vert \x_{[1,i]}, x_{i} = x_{i+1}} \Pr(x_i = x_{i+1}) \\
    &  & \quad + \E{f_n(\x) \middle\vert \x_{[1,i]}, x_{i} \ne x_{i+1}} \Pr(x_i \ne x_{i+1}) .
\end{eqnarray*}

By Lemma~\ref{lemma: f_mn partition}, for $1 < i < n$,
\begin{equation}\label{equ: Z_i bound}
    \begin{split}
        Z_i & \le \frac{1}{m} \cdot \mathbb{E} \left[f_{m,n}(\x_{[1,i]}) + f_{m,n}(\x_{[i+1,n]}) - mn+n+1 \vert x_i = x_{i+1}, \x_{[1,i]} \right] \\
            & \quad + \frac{m-1}{m} \cdot \mathbb{E} \left[ f_{m,n}(\x_{[1,i]}) + f_{m,n}(\x_{[i+1,n]})  \vert x_i \ne x_{i+1}, \x_{[1,i]}\right] \\
            & = f_{m,n}(\x_{[1,i]}) + g_{m,n}(i) + \frac{1}{m}(-mn + n + 1) \\
            & = f_{m,n}(\x_{[1,i]}) + g_{m,n}(i) - n + \frac{n}{m} + \frac{1}{m}, \\
        Z_i & \ge  \frac{1}{m} \cdot \mathbb{E} \left[f_{m,n}(\x_{[1,i]}) + f_{m,n}(\x_{[i+1,n]}) - mn+n+1 \vert x_i = x_{i+1}, \x_{[1,i]} \right] \\
            & \quad + \frac{m-1}{m} \cdot \mathbb{E} \left[ f_{m,n}(\x_{[1,i]}) + f_{m,n}(\x_{[i+1,n]} + 1 \right.\\
	    & \quad \left. - t(\x_{[1,i]})h(\x_{[i+1,n]}))  \vert x_i \ne x_{i+1}, \x_{[1,i]}\right] \\
            & = f_{m,n}(\x_{[1,i]}) + g_{m,n}(i) + \frac{1}{m}(- mn + n + 1) \\
	    & \quad + \frac{m-1}{m}\left(1 - t(\x_{[1,i]})\E{h(\x_{[i+1,n]}) \middle \vert x_{i} \ne x_{i+1}}\right) \\
             & = f_{m,n}(\x_{[1,i]}) + g_{m,n}(i) + 1 -n + \frac{n}{m} - \frac{m-1}{m}t(\x_{[1,i]})\left(2 - \frac{1}{m^{n-i-1}}\right),
    \end{split}
\end{equation}
where $g_{m,n}(i) := \E{f_{m,n}(\x_{[i+1,n]}) \middle \vert \x_{[1,i]}}$, and by Lemma~\ref{lemma: partial expectation}, we have
\begin{eqnarray}\label{equ: g_i}
  \lefteqn{g_{m,n}(i) = \underset{\x \in \mathbb{Z}_{m}^{n}}{\mathbb{E}}[f_{m,n}(\x_{[i+1,n]}) \vert \x_{[1,i]}]}\nonumber \\
  & = & \underset{\x \in \mathbb{Z}_{m}^{n-i}}{\mathbb{E}} [f_{m,n}(\x)] \nonumber \\
  & = & (mn - n - 1) \E{\rho(\x)} -\frac{1}{2} \E{\sum_{i=1}^{a(\x)} s_i^2} + \frac{3}{2} \E{\sum_{i=1}^{a(\x)} s_i} - \E{a(\x)} \nonumber \\
  & = & n(n-i)(m + \frac{1}{m} -2) - \frac{n}{m} + \frac{1}{m-1} - \frac{1}{(m-1)m^{n-i}} + i - \frac{1}{m^{n-i}}.
\end{eqnarray}

Now we claim that $|Z_i - Z_{i-1}|$ can be bounded as follows.
The proof details are left in Appendix~\ref{sec-appendixa}.

\begin{equation} \label{equ: bounded difference}
    |Z_i - Z_{i-1}| \le
    \begin{cases}
        0, &  i = 1, \\
        n(m + \frac{1}{m}),  & 2 \le i \le n. \\
    \end{cases}
\end{equation}

The result then follows by Theorem~\ref{thm: azuma}.
\end{proof}

\begin{remark}
    The key step of proving Theorem~\ref{thm: binary distribution} and Theorem~\ref{thm: m-ary distribution} is to calculate $Z_i$ to bound $|Z_i - Z_{i-1}|$.
    One may wonder why $m=2$ is not a special case of the general $m$.
    This is because we do not have explicit expressions of $Z_i$ for $m > 2$, but it is possible to derive it.
    Indeed, the proof of Lemma~\ref{lemma: f_mn partition} in fact expresses $f_{m,n}(\x)$ by $f_{m,n}(\x_{[1,i]})$, $f_{m,n}(\x_{[i+1,n]})$ and $l_i,r_i$, where $\x_{[i - l_i + 1, i + r_i]}$ is the maximal alternating segment that lies on both $\x_{[1,i]}$ and $\x_{[i+1,n]}$.
    Furthermore, one can calculate the conditional expectation of $l_i$ and $r_i$ given $\x_{[1,i]}, x_{i+1}, x_{i+2}$, which yields an explicit expression on $Z_i$ and $Z_i - Z_{i-1}$.
    However, this calculation is much more complex and neither significantly improves Theorem~\ref{thm: m-ary distribution}  (still, $|Z_i - Z_{i-1}| = O(n)$) nor gives more insights on this problem.
\end{remark}

\section{Conclusion}\label{sec: conclusion}
In this paper, we analyze the distribution of $|L_1(\x)|$ for $\x \in \mathbb{Z}_m^n$ by Azuma's inequality. The numerical result suggests that $|L_1(\x)|$ are more concentrated than we expected. Specifically, the gap between the simulation results and the bounds in Theorem~\ref{thm: binary distribution} and Theorem~\ref{thm: m-ary distribution} is still large (see Appendix~\ref{app:simulation}), leaving the derivation of better bounds as an open problem. Intuitively, the distribution of $|L_t(\x)|$ should be more and more concentrated as $t$ grows. For example, $|L_n(\x)| = m^n$ for all $\x \in \mathbb{Z}_m^n$. However, finding the distribution of $|L_t(\x)|$ is in general difficult and left open. Very recently, He and Ye~\cite{HY23} considered the case of radius two, and further gave the concentration bound for $|L_2(\x)|$.







\providecommand{\bysame}{\leavevmode\hbox to3em{\hrulefill}\thinspace}
\providecommand{\MR}{\relax\ifhmode\unskip\space\fi MR }
\providecommand{\MRhref}[2]{%
  \href{http://www.ams.org/mathscinet-getitem?mr=#1}{#2}
}
\providecommand{\href}[2]{#2}

\appendix

\section{Proof of Eq.~\eqref{equ: bounded difference}}\label{sec-appendixa}
We follow the notations in the proof of Theorem~\ref{thm: m-ary distribution}.

\begin{itemize}
    \item Case $i = 1$:
          By symmetry, we have $|Z_1 - Z_0| = 0$.

    \item Case $1< i <n$:

          By Eq.~\eqref{equ: Z_i bound}, we have
    \begin{eqnarray*}
      \lefteqn{Z_i - Z_{i-1}} \\
      & \ge & f_{m,n}(\x_{[1,i]}) + g_{m,n}(i) + 1 -n + \frac{n}{m} - \frac{m-1}{m}t(\x_{[1,i]})\left(2 - \frac{1}{m^{n-i-1}}\right) \\
      &  &  \quad -  \left[f_{m,n}(\x_{[1,i-1]}) + g_{m,n}(i-1) -n + \frac{n}{m} + \frac{1}{m}\right] \\
      & = &  f_{m,n}(\x_{[1,i]}) - f_{m,n}(\x_{[1,i-1]}) + g_{m,n}(i) - g_{m,n}(i-1) \\
      &  & \quad  + \frac{m-1}{m} - \frac{m-1}{m}t(\x_{[1,i]})\left(2 - \frac{1}{m^{n-i-1}}\right) ,
    \end{eqnarray*}

    \begin{eqnarray*}
      \lefteqn{Z_i - Z_{i-1}} \\
      & \le & f_{m,n}(\x_{[1,i]}) + g_{m,n}(i) - n + \frac{n}{m} + \frac{1}{m} \\
      &   & \quad - \left[ f_{m,n}(\x_{[1,i-1]}) + g_{m,n}(i-1) + 1 -n + \frac{n}{m} - \right. \\
      &   &  \quad \left. \frac{m-1}{m}t(\x_{[1,i-1]})\left(2 - \frac{1}{m^{n-i}}\right)\right] \\
      & = & f_{m,n}(\x_{[1,i]}) - f_{m,n}(\x_{[1,i-1]}) + g_{m,n}(i) - g_{m,n}(i-1)  \\
      &   &  \quad - \frac{m-1}{m} + \frac{m-1}{m}t(\x_{[1,i-1]})\left(2 - \frac{1}{m^{n-i}}\right).
    \end{eqnarray*}
          By Lemma~\ref{lemma: f_mn partition}, we have\footnote{Note that $f_{m,n}(x_i) = mn - n - 1$, and we then express $f_{m,n}(\x_{[1,i]})$ by $f_{m,n}(\x_{[1,i-1]})$ and $f_{m,n}(x_{i})$.}
        \[
            0 \le f_{m,n}(\x_{[1,i]}) - f_{m,n}(\x_{[1,i-1]}) \le mn - n - 1,
        \]
        and
        \[
          g_{m,n}(i) - g_{m,n}(i-1) = 1 - n(m + \frac{1}{m} - 2) - \frac{1}{m^{n-i}}.
        \]

          Therefore, we have
        \begin{equation*}
            \begin{split}
                Z_i - Z_{i-1} & \ge  1 - n \left(m + \frac{1}{m} - 2 \right) - \frac{1}{m^{n-i}} + \frac{m-1}{m} - \frac{m-1}{m}t(\x_{[1,i]})\left(2 - \frac{1}{m^{n-i-1}}\right) \\
              & >  - n\left(m + \frac{1}{m} - 2\right) - \frac{1}{m^{n-i}} - (n-1) \cdot 2 \\
              & >  - n\left(m + \frac{1}{m}\right) + 1,
            \end{split}
        \end{equation*}
        and also
        \begin{equation*}
            \begin{split}
                Z_i - Z_{i-1} & \le mn - n - 1 + 1 - n\left(m + \frac{1}{m} - 2\right) - \frac{1}{m^{n-i}}  - \frac{m-1}{m}  \\ 
		& \quad + \frac{m-1}{m}t(\x_{[1,i-1]})\left(2 - \frac{1}{m^{n-i}}\right) \\
                & = n\left(1 - \frac{1}{m}\right) - \frac{1}{m^{n-i}} - \frac{m-1}{m} + \frac{m-1}{m}t(\x_{[1,i-1]})\left(2 - \frac{1}{m^{n-i}}\right) \\
                & <  n\left(1 - \frac{1}{m}\right) + 2n = n\left(3 - \frac{1}{m}\right).
            \end{split}
        \end{equation*}
          Note that $|-n(m + \frac{1}{m}) + 1| \le n(m + \frac{1}{m})$ and $|n(3 - \frac{1}{m})| \le n(m + \frac{1}{m})$. Hence, we have $|Z_i - Z_{i-1}| \le n(m + \frac{1}{m})$ for $1 < i < n$.

    \item Case $i = n$:

          By Eq.~\eqref{equ: Z_i bound} and Lemma~\ref{lemma: f_mn partition}, we have
          \begin{equation*}
            0 \le f_{m,n}(\x) - f_{m,n}(\x_{[1,n-1]}) \le mn -n - 1,
          \end{equation*}
          and
          \begin{equation*}
            \begin{split}
                Z_n - Z_{n-1} & \le  f_{m,n}(\x) - \left[ f_{m,n}(\x_{[1,n-1]}) + \left(1 - \frac{1}{m} \right) \left(mn -n -t(\x_{[1,n-1]})\right)  \right]  \\
                & \le  mn - n - 1 - \left(1 - \frac{1}{m} \right) (mn - n -t(\x_{[1,n-1]})) \\
                & = t(\x_{[1,n-1]}) \left(1 - \frac{1}{m} \right) - \frac{n}{m} \\
                & < n \left( 1 - \frac{1}{m} \right),
            \end{split}
          \end{equation*}
          and also,
          \begin{equation*}
            \begin{split}
                Z_n - Z_{n-1} &  \ge f_{m,n}(\x) - \left[  f_{m,n}(\x_{[1,n-1]}) + \left(1 - \frac{1}{m}\right)(mn - n -1) \right] \\
                & \ge -\left(1 - \frac{1}{m}\right)(mn - n - 1)\\
                & = -n\left(m + \frac{1}{m} \right) - \frac{1}{m} + 1.
            \end{split}
          \end{equation*}
          Thus,  we have $|Z_n - Z_{n-1}| \le n(m + \frac{1}{m})$.
\end{itemize}

\section{Simulation results}~\label{app:simulation}

We independently pick $x \in \mathbb{Z}_m^n$ uniformly at random and record the value $|L_1(\x)|$.
The distribution of $|L_1(\x)|$ is then reflected by the frequency that $|L_1(\x)|$ lies in different intervals. We also compare it with the expected frequency given by the bounds in Theorems~\ref{thm: binary distribution} and~\ref{thm: m-ary distribution}.
For instance, the expected frequency of the event $|L_1(\x)| > \tau$ by Eq.~\eqref{equ: upper tail} is
$$N e^{-2 \left(\frac{\tau - \E{|L_1(\x)|}}{n \sqrt{n-1}}\right)^2},$$
where $N$ is the sample size.
The simulation results for $n = 100$, $m = 2,3,4,5$ are depicted in Fig.~\ref{fig:1}.

\begin{figure}[h]
    \includegraphics[width=\textwidth]{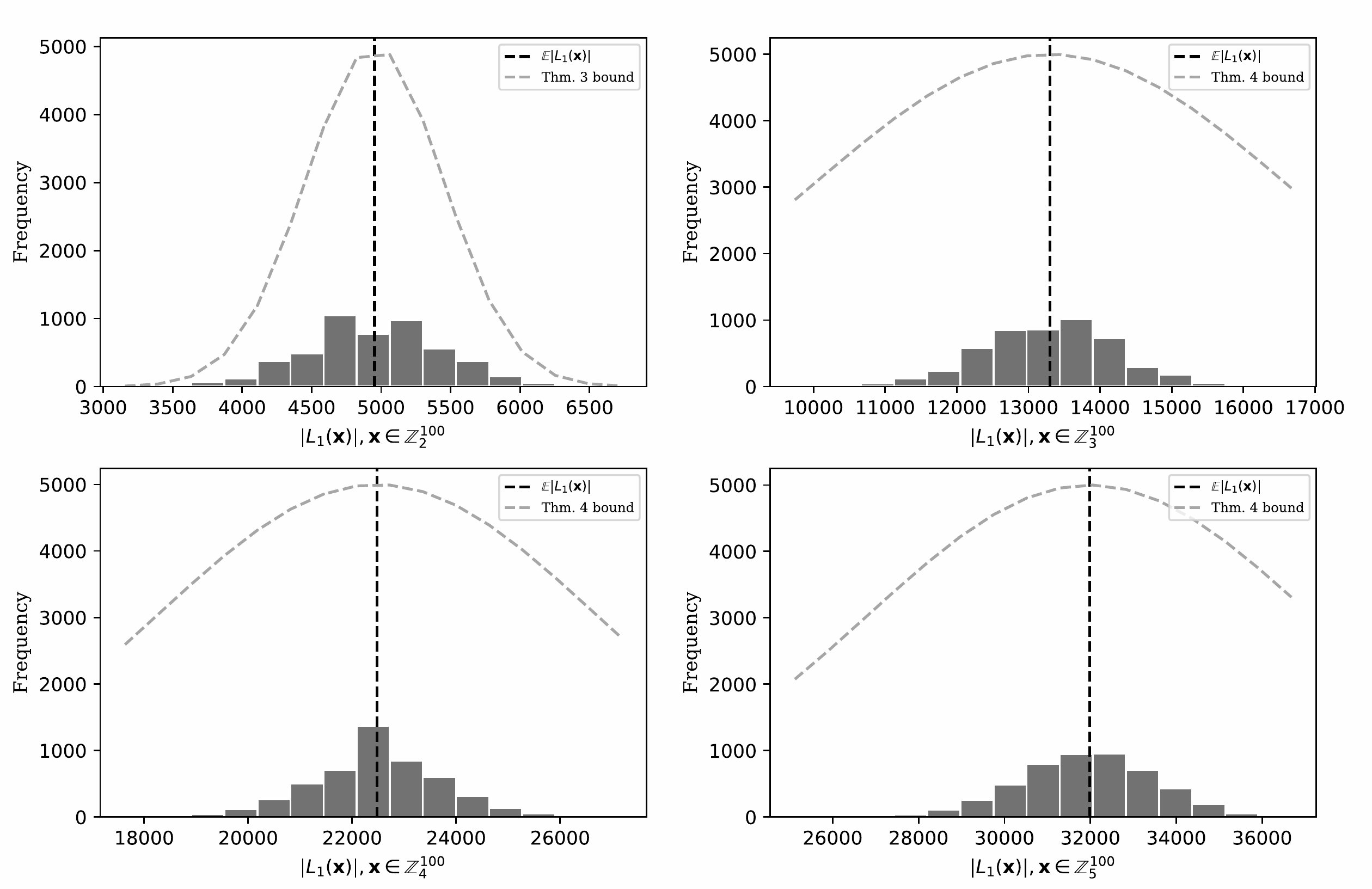}
    \caption{$N = 5000$ for each experiment.}
    \label{fig:1}
  \end{figure}

\end{document}